\documentclass[12pt,onecolumn]{./IEEEtran}
\usepackage{graphicx}
\usepackage{amssymb}
\usepackage{amsmath}

\newtheorem{theorem}{Theorem}
\newtheorem{algorithm}{Algorithm}

\newtheorem{corollary}{Corollary}
\newtheorem{example}{Example}

\newcommand{\Phib}{{\mbox{\boldmath$\Phi$}}}
\newcommand{\lcm}{{\rm lcm}}
\renewcommand{\baselinestretch}{2}

\begin{document}
\title{A Systematic Framework for the Construction of Optimal Complete Complementary Codes
}
\author{Chenggao~Han,~\IEEEmembership{Member,~IEEE,}
        	     Naoki~Suehiro,~\IEEEmembership{Senior Member,~IEEE,}
      and Takeshi~Hashimoto,~\IEEEmembership{Member,~IEEE,}
\thanks{Manuscript received March 26, 2007; revised February 18, 2009.}
\thanks{C. Han and T. Hashimoto are with Department of Electronic Engineering, Faculty of Electro-Communications, University of Electro-Communications, 1-5-1 Chofugaoka, Chofu-shi, Tokyo 185-8585, Japan (E-mail:\{hana,hasimoto\}@ee.uec.ac.jp).}
\thanks{N. Suehiro is with Graduate School of Information and System Engineering, University of Tsukuba, 1-1-1 Tennoudai, Tsukuba-shi, Ibaraki 305-0006, Japan (E-mail:suehiro@iit.tsukuba.ac.jp).}}

\maketitle
\begin{abstract}
The complete complementary code (CCC) is a sequence family with ideal correlation sums which was proposed by Suehiro and Hatori. Numerous literatures show its applications to direct-spread code-division multiple access (DS-CDMA) systems for inter-channel interference (ICI)-free communication with improved spectral efficiency. In this paper, we propose a systematic framework for the construction of CCCs based on $N$-shift cross-orthogonal sequence families ($N$-CO-SFs). We show theoretical bounds on the size of $N$-CO-SFs and CCCs, and give a set of four algorithms for their generation and extension. The algorithms are optimal in the sense that the size of resulted sequence families achieves theoretical bounds and, with the algorithms, we can construct an optimal CCC consisting of sequences whose lengths  are not only almost arbitrary but even variable between sequence families. We also discuss the family size, alphabet size, and lengths of constructible CCCs based on the proposed algorithms.
\end{abstract}


\section{Introduction}
In a multi-path environment, the performance of a direct-spread code-division multiple access (DS-CDMA) system relies on the correlation properties of the employed spreading sequences, and the full spectral efficiency \cite{Verdu99-03J} is attained only when the sequences have ideal correlation properties, {\it i.e.}, when the auto-correlation of every sequence is zero except for zero shift and the cross-correlation between every pair of sequences is zero for all shifts. Unfortunately, such a sequence set does not exist and, for practical DS-CDMA systems, we employ sequences whose correlations have as small side-lobes as possible, like Gold sequences, M-sequences, and Kasami sequences \cite{Fan99B}. Nonideal correlation properties of these sequences cause the near-far problem and inter-channel interference (ICI) that limit spectral efficiency of the DS-CDMA systems.

In \cite{Suehiro98-11C} and \cite{Suehiro00-09C}, Suehiro {\it et al.} have studied ICI-free DS-CDMA systems   based on the complete complementary code (CCC) proposed in \cite{Suehiro88-01J}. CCC is a collection of sequence sets, called a sequence family in this paper, with the property that the auto-correlation sum in each sequence set is zero except for zero shift and the cross-correlation sum between each pair of distinct sequence sets is zero for all shifts. (Exact definitions are given later.) In a CCC-based DS-CDMA (CCC-CDMA) system employing a CCC consisting of $M$ sequence sets of size $N$, the sequence sets are assigned to up to $M$ users and each user transmits $N$ spectrally spread signals through $N$ independent subchannels. At the receiver, these signals from the subchannels are passed through the corresponding matched filters and combined. The subchannels may be separated in frequency \cite{Suehiro98-11C} or in time \cite{Suehiro00-09C}. The former case was also analyzed by Tseng and Bell in \cite{Tseng00-01J}, while numerical results of the latter can be found in \cite{Chen06-02J,Han01-09C}, and \cite{Kojima06-09J}. Today, application of CCCs is also extended to multi-carrier systems \cite{Chen01-10J} and multi-input multi-output (MIMO) systems \cite{Chen07-02J,Lu08-01J}. In the area of sequence design, CCCs are used for constructing zero correlation zone (ZCZ) sequence sets \cite{Suehiro94-06J,Appuswamy06-08J,Han07-09C,Han08-12J},  which also provide ICI-free convolutional spreading CDMA systems \cite{Nalin06-11J,Nalin07-06J}, and are used in other areas such as image processing \cite{Khamy04-03C,Khamy05-11C}, synchronization \cite{Detert04-05C,Liu07-03C}, and signal processing \cite{Ozgur05-03C}.

The works leading to CCC started in 1961 when Golay \cite{Golay61-4J} studied a pair of binary sequences the sum of whose auto-correlations becomes zero except for zero shift and termed it a {\it complementary pair}. Following the Golay's work, properties of complementary pairs and relationships to other types of sequences were investigated by Turyn \cite{Turyn63-01J} and by Taki {\it et al.} \cite{Taki69-03J}. Tseng and Liu \cite{Tseng72-09J} extended the Golay's idea to a {\it complementary set}, a set of those sequences the sum of whose auto-correlations is zero except for zero shift and studied binary complementary sets with orthogonal properties. Multiphase complementary pairs and complementary sets were studied by Sivaswamy \cite{Sivaswamy78-09J} and Frank \cite{Frank80-11J}, respectively. Since sequences from complementary pairs can be used to reduce the peak-to-average power ratio (PAPR) of orthogonal frequency division multiplexing (OFDM) systems \cite{Davis99-11J,Nee00B}, complementary pairs consisting of symbols from quadrature amplitude modulation (QAM) and polyphase signal sets have been also investigated by numerous researchers \cite{Robing01-07J,Tarokh03-01J,Lee06-04J,Li05-03J,Fiedler06-09J}. Differently from these works, Suehiro and Hatori \cite{Suehiro88-01J} extended the idea of the complementary set to what is now known as the {\it complete complementary code} (CCC), the collection of complementary sets with the additional property that the cross-correlation sum between every pair of distinct sequence sets is zero. In \cite{Suehiro88-01J}, they introduced {\it $N$-shift cross-orthogonal sequence sets} ($N$-CO-SSs) and constructed CCCs from the $N$-CO-SSs. However, their construction has a restriction that the sequence length of the constructed CCC must be no shorter than $N^2$ when the CCC consists of $N$ complementary sets, which limits the spectral efficiency of the CCC-CDMA systems discussed in \cite{Suehiro98-11C} and \cite{Suehiro00-09C}. To overcome this weakness, we have proposed improved construction methods for CCCs with length $N/P$ in \cite{Han04-04J} and with length $MN/P$ in \cite{Han04-10C} for arbitrary integers $1 \leq P,M \leq N$. We have also proposed a method for the construction of CCCs which allows different complementary sets to have distinct sequence lengths in \cite{Raja07-07J}. 
Although there are other approaches to the construction of CCCs, {\it e.g.}, from existing CCC \cite{Huang02-11C,Lu04-12J,Marziani07-05J,Yang07-11C,Wu08-01J} or from Reed-Muller codes \cite{Chen08-11J}, CCCs constructed by these methods lack the flexibility of allowing variable length sequences that the one in \cite{Raja07-07J} has. Recently, investigation of CCC has been extended further, for example, to the periodic cases \cite{Torii00-07C,Torii01-03C} and to the multi-dimensional cases \cite{Farkas03-10C,Turcsany04-10C,Zhang05-02J,Raja06-10C}.

In this paper, extending the notation of $N$-CO-SS, we introduce {\it $N$-shift cross-orthogonal sequence families} ($N$-CO-SFs) whose sequences may have distinct lengths, and we propose methods to construct $N$-CO-SFs and CCCs from those $N$-CO-SFs. We prove theoretical limitations on the size of $N$-CO-SF and of CCC and present general and systematic methods to construct the optimal $N$-CO-SFs and CCCs. These constructions of $N$-CO-SFs and CCCs consist of code generation and extension, and only unitary(-like) matrices are used. We also discuss the family size, alphabet size, and sequence length of the constructed CCCs.

The rest of this paper is as follows. In Section II, after the definition of the correlation and correlation sum, we define complementary set, CCC, and $N$-CO-SF. In Section III, we derive theoretical bounds on the size of an $N$-CO-SF and introduce generation/expansion algorithms for $N$-CO-SF with the notion of multi-level partition. Similar to Section III, Section IV gives theoretical bounds on the size of a CCC and generation/expansion algorithms for CCC. The family size, alphabet size, and lengths of the CCCs constructed by our algorithms are discussed in Section V. Section VII gives conclusions.

\subsection{Notations}
For two nonnegative integers $a$ and $b$, the remainder and quotient of $a/b$ are denoted by $[a]_b$ and $\lfloor a/b \rfloor$, respectively, and, ${\rm lcm}(a,b)$ denotes their least common multiple. 
A vector is denoted by a bold lowercase letter and is also represented with its entries as ${\bf v} = (v_n)_{n=0}^{N-1}$. For simplicity, we identify the vector $\bf v$ and a sequence $v(n)$ satisfying $v(n) = v_n$ for $0 \leq n < N$ and $v(n)=0$ otherwise. The concatenation of $M$ vector ${\bf v}_m$, $0 \leq m < M$, is denoted by the $({\bf v}_0~{\bf v}_1~\cdots~{\bf v}_{M-1})$ or $({\bf v}_m)_{m=0}^{M-1}$ and ${\bf 0}_N$ denotes all zero vector of length $N$.
A matrix is denoted by a bold uppercase letter and an $M\times N$ matrix ${\bf A}$ with entries $a_n^m$, $0 \leq m < M$ and $0 \leq n < N$, is represented with its entries as ${\bf A}=[a_n^m]_{m=0,n=0}^{M-1,N-1}$. The $m$th row and the $n$th column vectors of $\bf A$ are denoted by ${\bf a}^m$ and ${\bf a}_n$, respectively.
Moreover, we let ${\bf A}^\ast$ and ${\bf A}^H$ denote the complex conjugate and Hermitian transpose of $\bf A$, respectively. 
An indexed set is a set of numbered elements and is denoted by outline letters as $\mathbb{S}=\{s_n\}_{n=0}^{N-1}$. $|\mathbb{S}|$ denotes the size of $\mathbb{S}$. Furthermore, in the notations for vectors, matrices, and sets, we occasionally omit the range of indices when it is obvious, {\it e.g.}, $\{s_n\}_{n=0}^{N-1}=\{s_n\}$.

\section{Definitions}

For two sequences ${\bf s}$ and ${\bf s}^\prime$ with lengths $L$ and $L^\prime$, respectively, the {\it aperiodic correlation}, or {\it correlation} simply, is given by
 \begin{eqnarray}
 R_{ {\bf s}, {\bf s}^\prime} ( \tau ) := \sum_{l=0}^{L-1} s(l)[s^\prime( l+\tau )]^\ast
 \label{corrfun}
 \end{eqnarray}
If ${\bf s} \neq {\bf s}^\prime$, it is the {\it cross-correlation} of ${\bf s}$ and $\bf s^\prime$ and, if ${\bf s} = {\bf s}^\prime$, it is the {\it auto-correlation} of ${\bf s}$ and is simply denoted by $R_{ \bf s} ( \tau )$. $E_{\bf s}:=R_{\bf s}(0)$ is the energy of $\bf s$.
%
When $L^\prime = L$, the {\it periodic correlation} $\tilde{R}_{ {\bf s}, {\bf s}^\prime}$ is given by the same expression as (\ref{corrfun}) except that $l+\tau$ is replaced by $[l+\tau]_L$ in (\ref{corrfun}).
Between the periodic and aperiodic definitions, we note the identity
\begin{eqnarray}
\tilde{R}_{{\bf s},{\bf s}^\prime}(\tau) = R_{{\bf s},{\bf s}^\prime}(\tau) +[1-\delta(\tau)]R_{{\bf s},{\bf s}^\prime}(\tau-L) 
\label{relcorr}
\end{eqnarray}

We call an indexed set $\mathbb{S}$ consisting of $N$ sequences of length $L$ an $(N,L)$-{\it sequence set} (SS) and call $L$ the length of $\mathbb{S}$. 
For given $(N,L)$-SS $\mathbb{S}=\{{\bf s}_n\}$ and $(N,L^\prime)$-SS $\mathbb{S}^\prime=\{{\bf s}^\prime_n\}$, we define the {\it correlation sum} as
\begin{eqnarray}
{\cal R}_{\mathbb{S},\mathbb{S}^\prime}(\tau) := \sum_{n=0}^{N-1}R_{{\bf s}_n,{\bf s}_n^\prime}(\tau)
\label{Csum}
\end{eqnarray}
If $\mathbb{S} \neq \mathbb{S}^\prime$, ${\cal R}_{\mathbb{S},\mathbb{S}^\prime}(\tau)$ is called the {\it cross-correlation sum} of $\mathbb{S}$ and $\mathbb{S}^\prime$  and, if $\mathbb{S} = \mathbb{S}^\prime$, it is called the {\it auto-correlation sum} of $\mathbb{S}$ and denoted by ${\cal R}_{\mathbb{S}}(\tau)$. We call $E_{\mathbb{S}}={\cal R}_{\mathbb{S}}(0)$ the energy of $\mathbb{S}$. The {\it periodic correlation sum} $\tilde{\cal R}_{\mathbb{S},\mathbb{S}^\prime}(\tau)$ is also introduced in a similar manner.

An $(N,L)$-SS $\mathbb{S}$ is called a {\it complementary set} (CS) and denoted by $(N,L)$-CS if the auto-correlation sum of $\mathbb{S}$ is zero except for zero-shift, {\it i.e.},
\begin{eqnarray}
{\cal R}_{\mathbb{S}}(\tau) = E_{\mathbb{S}}\delta(\tau)
\label{CSfun}
\end{eqnarray}
where $\delta({\bf a})$ is {\it Kronecker's delta function}.
\begin{example}
Assuming the convention that the signs $\pm$ mean $\pm1$, we let ${\bf s}_0^0=(+++-)$ and ${\bf s}_1^0= (+-++)$. Then the auto-correlation sum of $\mathbb{S}^0=\{{\bf s}_0^0,{\bf s}_1^0\}$ gives, in a vector form,
\begin{eqnarray}
\left({\cal R}_{\mathbb{S}^0}(\tau)\right)_{\tau=-3}^{3}&= & \left(R_{{\bf s}_0^0}(\tau)+R_{{\bf s}_1^0}(\tau)\right)_{\tau=-3}^{3}\nonumber\\
& = & (-1,0,1,4,1,0,-1)+(1,0,-1,4,-1,0,1)\nonumber\\
& = & (0,0,0,8,0,0,0)
\label{CSE}
\end{eqnarray}
Hence, $\mathbb{S}_0$ is a $(2,4)$-CS.
\end{example}

We introduce a collection of SSs. Given $(N,L^{(m)})$-SSs $\mathbb{S}^m=\{{\bf s}_n^m\}_{n=0}^{N-1}, m = 0,1,\cdots,M-1$, we call the set ${\cal S}=\{\mathbb{S}^m\}_{m=0}^{M-1}$ an $\left(M,N,\mathbb{L}\right)$-{\it sequence family} (SF) with family size $M$ and length set $\mathbb{L} := \sqcup\{L^{(m)}\}_{m=0}^{M-1}$, where $\sqcup\{L^{(m)}\}_{m=0}^{M-1}$ denotes the smallest set that contains all $L^{(m)}$, $m = 0,1,\cdots, M-1$, {\it i.e.}, the set of all the distinct lengths in $\{L^{(m)}\}_{m=0}^{M-1}$. For the length set $\mathbb{L}$, let $\max(\mathbb{L})$ be the maximum value in $\mathbb{L}$. Although sequences in an SS are considered to have the same length, we allow different SSs in an SF to have different lengths. The SF is also represented by an $M\times N$ matrix form $\mathcal{S} =\left[{\bf s}_n^m\right]$ with sequences as its entries. 

An $\left(M,N,\mathbb{L}\right)$-SF $\mathcal{C}=\{\mathbb{C}^m\}$ is called a {\it complete complementary code} (CCC) and denoted by $(M,N,\mathbb{L})$-CCC if each $\mathbb{C}^m$ is an $(N,L^{(m)})$-CS with $L^{(m)}\in \mathbb{L}$ and every pair $\mathbb{C}^m,\mathbb{C}^{m^\prime} \in {\cal C}$ satisfies
\begin{eqnarray}
{\cal R}_{\mathbb{C}^m,\mathbb{C}^{m^\prime}}(\tau) = E_{\mathbb{C}^m}\delta(m-{m^\prime})\delta(\tau)
\label{CCCfun}
\end{eqnarray}

\begin{example}
In addition to the $(2,4)$-CS $\mathbb{S}^0$ in Example 1, we let $\mathbb{S}^1 = \{{\bf s}_0^1,{\bf s}_1^1\}$ with ${\bf s}^1_0=(++-+)$ and ${\bf s}^1_1= (+---)$, respectively. Then, $\mathbb{S}^1$ is also a $(2,4)$-CS since the auto-correlation sum of $\mathbb{S}^1$ gives
\begin{eqnarray}
\left({\cal R}_{\mathbb{S}^1}(\tau)\right)_{\tau = -3}^{3}& = & \left(R_{{\bf s}^1_0}(\tau)+R_{{\bf s}^1_1}(\tau)\right)_{\tau = -3}^{3}\nonumber\\
& = & (1,0,-1,4,-1,0,-1)+(-1,0,1,4,1,0,-1)\nonumber\\
& = & (0,0,0,8,0,0,0)
\label{CCCE1}
\end{eqnarray}
and the cross-correlation sum of $\mathbb{S}^0 $ and $\mathbb{S}^1$ is given by
\begin{eqnarray}
\left(\mathcal{R}_{\mathbb{S}^0,\mathbb{S}^1}(\tau)\right)_{\tau = -3}^{3}& = & \left(R_{{\bf s}^0_0,{\bf s}^1_0}(\tau)+R_{{\bf s}^0_1,{\bf s}^1_1}(\tau)\right)_{\tau = -3}^{3}\nonumber\\
& = & (-1,0,3,0,1,0,1)+(1,0,-3,0,-1,0,-1)\nonumber\\
& = & (0,0,0,0,0,0,0)
\label{CCCE2}
\end{eqnarray}
Therefore, $\mathcal{S} =\{\mathbb{S}^0,\mathbb{S}^1\}$  is a $(2,2,\{4\})$-CCC. It can be represented by a matrix form
\begin{eqnarray}
{\cal S} = \left[
\begin{array}{cc}
{\bf s}^0_0 & {\bf s}^0_1 \cr
{\bf s}^1_0 & {\bf s}^1_1
\end{array}
\right]
\label{Exp2}
\end{eqnarray}

\end{example}

We identify two SSs if they are the same except for sequence indexing, {\it e.g.}, $\{{\bf s}_0,{\bf s}_1\} = \{{\bf s}_0^\prime,{\bf s}_1^\prime\}$ if ${\bf s}_0^\prime = {\bf s}_1$ and ${\bf s}_1^\prime = {\bf s}_0$, and we identify two SFs if these are the same except for indexing of sequences and sets. For example, we identify $\left[
\begin{array}{cc}
{\bf s}^0_0 & {\bf s}^0_1 \cr
{\bf s}^1_0 & {\bf s}^1_1
\end{array}
\right],\left[
\begin{array}{cc}
{\bf s}^0_1 & {\bf s}^0_0 \cr
{\bf s}^1_1 & {\bf s}^1_0
\end{array}
\right]$, and $
\left[
\begin{array}{cc}
{\bf s}^1_0 & {\bf s}^1_1 \cr
{\bf s}^0_0 & {\bf s}^0_1
\end{array}
\right]$ with ${\cal S}$ in (\ref{Exp2}), but $
\left[
\begin{array}{cc}
{\bf s}^0_0 & {\bf s}^0_1 \cr
{\bf s}^1_1 & {\bf s}^1_0
\end{array}
\right]\neq {\cal S}$.

In \cite{Suehiro88-01J}, $N$-CO-SSs are introduced as materials to construct CCCs. In our framework, we take SFs as materials to construct CCCs since SF allows component SSs of distinct lengths. We say that a shift (either aperiodic or periodic) is an {\it $N$-shift} if it is a shift in $lN$ elements for an integer $l$. We call an $(M,1,\mathbb{L})$-SF ${\cal S}$ an {\it $N$-shift cross-orthogonal sequence family} ($N$-CO-SF) and write $(M,1,\mathbb{L})$-$N$-CO-SF if the auto-correlation sum of each SS (consisting of just one sequence) in $\cal S$ vanishes for all $N$-shifts except for zero-shift and the cross-correlation sum of each pair of SSs $\{{\bf s}^m\}$ and $\{{\bf s}^{m^\prime}\}$ in $\cal S$ vanishes for all $N$-shifts. The $\cal S$ consists of $M$ SSs each including just one sequence and is represented by an $M \times 1$ matrix. We note that the length of each SS of an $(M,1,\mathbb{L})$-$N$-CO-SF is divisible by $N$. It is not difficult to confirm that the sequences in Example 2 give two $(2,1,\{4\})$-$2$-CO-SFs ${\cal S}_0 = \left[
\begin{array}{c}
{\bf s}^0_0 \cr
{\bf s}^1_0
\end{array}
\right]$ and ${\cal S}_1 = \left[
\begin{array}{c}
{\bf s}^0_1 \cr
{\bf s}^1_1
\end{array}
\right]$. 

In the followings, we construct CCCs whose row consists of CSs and whose column consists of $N$-CO-SFs. There are two approach to construct such CCCs: 1) Construct CS (row) first and extend them to CCC and 2) construct $N$-CO-SFs (columns) first and extend to a CCC. Our approach shown in this paper belongs to the later.

\section{$N$-shift cross-orthogonal sequence family}
In this section, we first give a theoretical upper bound on the family size of an $N$-CO-SF and present generation and elongation methods to find the optimal $N$-CO-SFs in the sense of the bound. 

\subsection{An upper bound for $N$-CO-SF and a generation algorithm}

The following theorem is proved in Appendix I (see also \cite{Han07-03J}).
\begin{theorem}
Every $(M,1,\mathbb{L})$-$N$-CO-SF satisfies the bound $M \leq N$ and this bound is attainable.
\end{theorem}

Because of the theorem, we call an $(M,1,\mathbb{L})$-$N$-CO-SF {\it optimal} if $M = N$ holds. 

We next consider construction of the optimal $N$-CO-SFs. Our construction begin with unitary-like matrices. Let us call an $N \times N$ matrix ${\bf U}_N$ a {\it unitary-like matrix} if it satisfies ${\bf U}_N {\bf U}_N^H = {\bf U}_N^H {\bf U}_N = \alpha {\bf I}_N$ for some $\alpha > 0$, where ${\bf I}_N$ denotes the $N\times N$ identity matrix. Obviously, ${\bf U}_N$ is a unitary matrix if $\alpha = 1$. An example of the unitary-like matrix is the $N$-dimensional discrete Fourier transform (DFT) matrix which is denoted by ${\bf F}_N=[W_N^{mn}]_{m=0,n=0}^{N-1,N-1}$ with $W_N=\exp(-2j\pi/N)$. Another one is the $N$-dimensional Walsh-Hadamard matrix ${\bf H}_N$ defined recursively
\begin{eqnarray*}
{\bf H}_{2^m} := \left[
\begin{array}{cc}
{\bf H}_{2^{m-1}} & {\bf H}_{2^{m-1}} \cr
{\bf H}_{2^{m-1}} & -{\bf H}_{2^{m-1}} 
\end{array}
\right]
\end{eqnarray*}
with 
${\bf H}_{1} = \left[~1~\right]$. let us denote the $m$th row vector of ${\bf F}_N$ and ${\bf H}_N$ by ${\bf f}_N^m$ and ${\bf h}_N^m$, respectively. We note that rows of an $N \times N$ unitary-like matrix ${\bf U}_N$ constitute an $(N,1,\{N\})$-$N$-CO-SF and, if we consider each row as a collection of length-$1$ sequences, ${\bf U}_N$ can be regarded as an $(N,N,\{1\})$-CCC.

Now, before the description of a general construction method, we consider an example of a $2$-CO-SF. 
\begin{example}
Let $\mathbb{A}=\{{\bf h}_2^m\}_{m=0}^{1}$ be a $(2,2)$-SS consisting of rows of ${\bf H}_2$ and consider a $2\times 2$ unitary-like matrix ${\bf U}_{2} = {\bf H}_2$. Then, we can generate an optimal $(2,1,\{4\})$-$2$-CO-SF as
\begin{eqnarray}
\mathcal{S} &=& \left[
\begin{array}{c}
{\bf s}^{0}\cr
{\bf s}^{1}
\end{array}
\right]
= \left[
\begin{array}{c}
(u_0^0{\bf h}_2^0~u_1^{0}{\bf h}_2^1)\cr
(u_0^{1}{\bf h}_2^0~u_1^{1}{\bf h}_2^1)
\end{array}
\right]
= \left[
\begin{array}{c}
(+++-)\cr
(++-+)
\end{array}
\right]
\label{OSD}
\end{eqnarray}
where $u_n^m$ denotes the $(m,n)$th entry of ${\bf U}_N$. 
\end{example}

For a vector ${\bf v}=(v_n)_{n=0}^{N-1}$ and an $(M,L)$-SS $\mathbb{A}=\{{\bf a}_m\}_{m=0}^{M-1}$ and for $K={\rm lcm}(M,N)$, let $\odot$ be the connection operator to construct a sequence of length $KL$ as
\begin{eqnarray}
{\bf s}={\bf v}\odot\mathbb{A} := \left(v_{[k]_N}{\bf a}_{[k]_{M}}\right)_{k=0}^{K-1}
\label{OD}
\end{eqnarray}
Then, the construction (\ref{OSD}) can be expressed as $\mathcal{S} = \left[
\begin{array}{c}
{\bf u}_{2}^0\odot\mathbb{A}\cr
{\bf u}_{2}^1\odot\mathbb{A}
\end{array}
\right]$, where ${\bf u}_{2}^m$ denotes the $m$th row of ${\bf U}_{2}$\footnote{As shown in this example, in fact, if an SS $\mathbb{A}$ consists of rows of a $N\times N$ unitary-like matrix is given, an $N$-CO-SS can be constructed by performing $\odot$ operation between rows of ${\bf U}_{|\mathbb{A}|}$ and $\mathbb{A}$.}. Our goal is to construct an $N$-CO-SF. To this end, we introduce multi-level partition and generalize the construction given in Example 3.

Let us consider a $Q$-level partition of a set $\mathbb{A}$, where $\mathbb{A}$ is partitioned into $P_1$ subsets $\mathbb{A}^{(p_1)},0 \leq p_1 < P_1$, at the 1st partition level and, for each ${\bf p}_1$, the subset $\mathbb{A}^{(p_1)}$ is further partitioned into $P_2^{(p_1)}$ subsets $\mathbb{A}^{(p_1,p_2)}, 0 \leq p_2 < P^{(p_1)}_2$, at the 2nd partition level and so on. In the resultant partition tree, a subset obtained at the $q$th partition level is indexed by a path vector ${\bf p}_q=(p_1,p_2,\cdots,p_{q})$ as $\mathbb{A}^{({\bf p}_q)}$. Let $\mathbb{P}_q$ be the path vector set consisting of all path vectors ${\bf p}_q$ leading to subsets at the $q$th partition level.

\begin{example}
Let $\mathbb{A}=\{{\bf f}_6^m\}_{m=0}^{5}$ be a $(6,6)$-SS consisting of rows of ${\bf F}_6$. We first consider 1-level partition of $\mathbb{A}$ into $\mathbb{A}^{(0)} = \{{\bf f}_6^{m}\}_{m=0}^{1}$ and $\mathbb{A}^{(1)} = \{{\bf f}_6^{m}\}_{m=2}^{5}$. Next, let  ${\bf U}_{|\mathbb{A}^{(0)}|}^{(0)} ={\bf H}_2$ and ${\bf U}_{|\mathbb{A}^{(1)}|}^{(1)}={\bf H}_4$ be two unitary-like matrices whose dimensions are equal to $|\mathbb{A}^{(0)}| = 2$ and $|\mathbb{A}^{(1)}| = 4$, respectively. Then, we can generate an optimal $(6,1,\{12,24\})$-$6$-CO-SF
\begin{eqnarray}
\mathcal{S} &=& \left[
\begin{array}{c}
{\bf s}^{(0,0)}\cr
{\bf s}^{(0,1)}\cr
{\bf s}^{(1,0)}\cr
{\bf s}^{(1,1)}\cr{\bf s}^{(1,2)}\cr
{\bf s}^{(1,3)}
\end{array}
\right]
= \left[
\begin{array}{c}
{\bf u}_{|\mathbb{A}^{(0)}|}^{(0),0}\odot\mathbb{A}^{(0)}\cr
{\bf u}_{|\mathbb{A}^{(0)}|}^{(0),1}\odot\mathbb{A}^{(0)}\cr
{\bf u}_{|\mathbb{A}^{(1)}|}^{(1),0}\odot\mathbb{A}^{(1)}\cr
{\bf u}_{|\mathbb{A}^{(1)}|}^{(1),1}\odot\mathbb{A}^{(1)}\cr
{\bf u}_{|\mathbb{A}^{(1)}|}^{(1),2}\odot\mathbb{A}^{(1)}\cr
{\bf u}_{|\mathbb{A}^{(1)}|}^{(1),3}\odot\mathbb{A}^{(1)}
\end{array}
\right]
=\left[
\begin{array}{c}
{\bf h}_2^{0}\odot\mathbb{A}^{(0)}\cr
{\bf h}_2^{1}\odot\mathbb{A}^{(0)}\cr
{\bf h}_4^{0}\odot\mathbb{A}^{(1)}\cr
{\bf h}_4^{1}\odot\mathbb{A}^{(1)}\cr
{\bf h}_4^{2}\odot\mathbb{A}^{(1)}\cr
{\bf h}_4^{3}\odot\mathbb{A}^{(1)}
\end{array}
\right]
= \left[
\begin{array}{l}
(+{\bf f}_6^{0}~+{\bf f}_6^{1})\cr
(+{\bf f}_6^{0}~-{\bf f}_6^{1})\cr
(+{\bf f}_6^{2}~+{\bf f}_6^{3}~+{\bf f}_6^{4}~+{\bf f}_6^{5})\cr
(+{\bf f}_6^{2}~-{\bf f}_6^{3}~+{\bf f}_6^{4}~-{\bf f}_6^{5})\cr
(+{\bf f}_6^{2}~+{\bf f}_6^{3}~-{\bf f}_6^{4}~-{\bf f}_6^{5})\cr
(+{\bf f}_6^{2}~-{\bf f}_6^{3}~-{\bf f}_6^{4}~+{\bf f}_6^{5})
\end{array}
\right]
\label{NCOSexam1}
\end{eqnarray}
where ${\bf u}_{|\mathbb{A}^{(p)}|}^{(p),m}$ denotes the $m$th row of ${\bf U}_{|\mathbb{A}^{(p)}|}^{(p)}$.
\end{example}

In general, if an $(N,N)$-SS $\mathbb{A}$ consisting of rows of a unitary-like matrix ${\bf U}_N$ is given, an optimal $N$-CO-SF can be generated by the following algorithm.
\begin{algorithm}
~
\begin{enumerate}
\item  {\it Partition}: Perform arbitrary $1$-level partition of $\mathbb{A}$ and let $\left\{\mathbb{A}^{({\bf p}_1)}\right\}_{{\bf p}_1 \in \mathbb{P}_1}$, $\mathbb{P}_1=\left\{(p_1)\right\}_{p_1=0}^{P_1-1}$, be the resultant partition.

\item {\it Specification}: For each subset $\mathbb{A}^{({\bf p}_1)}$, specify an arbitrary $|\mathbb{A}^{({\bf p}_1)}|\times |\mathbb{A}^{({\bf p}_1)}|$ unitary-like matrix ${\bf U}_{|\mathbb{A}^{({\bf p}_1)}|}^{({\bf p}_1)}$.

\item {\it Connection operation}: Let $\mathbb{P} = \{{\bf p} = ({\bf p}_1,m); 0 \leq m <|\mathbb{A}^{({\bf p}_1)}|,{\bf p}_1 \in \mathbb{P}_1\}$ and generate an $(N,1,\mathbb{L})$-SF by
\begin{eqnarray}
{\cal S} = \left\{\left\{{\bf s}^{\bf p}\right\}\right\}_{{\bf p} \in \mathbb{P}}
=  \left\{\left\{{\bf u}_{|\mathbb{A}^{({\bf p}_1)}|}^{({\bf p}_1),m}\odot\mathbb{ A}^{({\bf p}_1)}\right\}\right\}_{({\bf p}_1,m) \in \mathbb{P}}
\label{NCOSgen}
\end{eqnarray}
and $\mathbb{L} = \sqcup\{|\mathbb{A}^{({\bf p}_1)}|N\}_{{\bf p}_1 \in \mathbb{P}_1}$.
\end{enumerate}
\end{algorithm}

The following theorem is proved in Appendix II. 
\begin{theorem}
The $(N,1,\mathbb{L})$-SF $\mathcal{S}$ constructed by Algorithm 1 is an optimal $N$-CO-SF.
\end{theorem}

Although Theorem 2 provides us a method to generate an optimal $N$-CO-SF, the sequence lengths $|\mathbb{A}^{({\bf p}_1)}|N$ are no larger than $N^2$ since $\mathbb{A}^{({\bf p}_1)}$ is a subset of the $(N,N)$-SS $\mathbb{A}$. 
To construct a longer $N$-CO-SF, we need an elongation algorithm discussed next.

\subsection{Elongation of an $N$-CO-SF}

We begin with an example again. 
 \begin{example}
 Let us consider elongation of the $(6,1,\{12,24\})$-$6$-CO-SF $\mathcal{S}$ constructed in Example 4. Now, we let $\mathbb{A}$ be the set of all the sequences in ${\cal S}$, {\it i.e.}, $\mathbb{A} = \{{\bf s}^{(0,0)},{\bf s}^{(0,1)},{\bf s}^{(1,0)},{\bf s}^{(1,1)},{\bf s}^{(1,2)},{\bf s}^{(1,3)}\}$. Compared with the generation algorithm, the elongation algorithm performs two-level partition: At the $1$st partition level, $\mathbb{A}$ is divided into two SSs according to the lengths of sequences, and each of the resultant SSs is arbitrarily divided at the $2$nd partition level. Thus, from $\left|\{12,24\}\right|=2$, we have subsets $\mathbb{A}^{(0)} = \left\{{\bf s}^{(0,m)}\right\}_{m=0}^{1}$ consisting of the sequences of length 12 and  $\mathbb{A}^{(1)} = \left\{{\bf s}^{(1,m)}\right\}_{m=0}^{3}$ consisting of the sequences of length 24 at the $1$st partition level. 
At the $2$nd partition level, we may let $\mathbb{A}^{(0,0)} = \mathbb{A}^{(0)} = \{{\bf s}^{(0,m)}\}_{m=0}^{1}$ and may divide $\mathbb{A}^{(1)}$ into two subsets: $\mathbb{A}^{(1,0)} = \{{\bf s}^{(1,m)}\}_{m=0}^{1}$ and $\mathbb{A}^{(1,1)}=\{{\bf s}^{(1,m)}\}_{m=2}^{3}$. The path vector set of such partition is $\mathbb{P}_2 = \{(0,0),(1,0),(1,1)\}$. Next, for each ${\bf p}_2 \in \mathbb{P}_2$, we specify an $|\mathbb{A}^{({\bf p}_2)}|$-CO-SFs. In this example, since $|\mathbb{A}^{({\bf p}_2)}|=2$ for all ${\bf p}_2\in \mathbb{P}_2$, we need three $2$-CO-SFs. Here, we assume the following $2$-CO-SFs:  $(2,1,\{2\})$-$2$-CO-SFs $\mathcal{V}^{(0,0)}$ and $\mathcal{V}^{(1,0)}$ both consisting of the rows of a unitary-like matrix ${\bf H}_2$, and a $(2,1,\{4\})$-$2$-CO-SF $\mathcal{V}^{(0,0)}$ given in Example 3, {\it i.e.},
\begin{eqnarray*}
\mathcal{V}^{(0,0)} = \mathcal{V}^{(1,0)}= \left[
\begin{array}{c}
(++) \cr
(+-)
\end{array}
\right];~
\mathcal{V}^{(1,1)} = \left[
\begin{array}{c}
(+++-)\cr
(++-+)
\end{array}
\right]
\end{eqnarray*}
Let ${\bf v}^{(p_1,p_2,m)}$ be the sequence in the $m$th SS of ${\cal V}^{(p_1,p_2)}$ (there is just one for each $m$) and let $\cal S$ is be a $(6,1,\{24,48,96\})$-$6$-CO-SF given by
\begin{eqnarray*}
\mathcal{S}&=& \left[
\begin{array}{c}
{\bf v}^{(0,0),0}\odot\mathbb{A}^{(0,0)}\cr
{\bf v}^{(0,0),1}\odot\mathbb{A}^{(0,0)}\cr
{\bf v}^{(1,0),0}\odot\mathbb{A}^{(1,0)}\cr
{\bf v}^{(1,0),1}\odot\mathbb{A}^{(1,0)}\cr
{\bf v}^{(1,1),0}\odot\mathbb{A}^{(1,1)}\cr
{\bf v}^{(1,1),1}\odot\mathbb{A}^{(1,1)}
\end{array}
\right]
=
\left[ 
\begin{array}{l}
(+{\bf s}^{(0,0)}~+{\bf s}^{(0,1)})\cr
(+{\bf s}^{(0,0)}~-{\bf s}^{(0,1)})\cr
(+{\bf s}^{(1,0)}~+{\bf s}^{(1,1)})\cr
(+{\bf s}^{(1,0)}~-{\bf s}^{(1,1)})\cr
(+{\bf s}^{(1,2)}~+{\bf s}^{(1,3)}~+{\bf s}^{(1,2)}~-{\bf s}^{(1,3)})\cr
(+{\bf s}^{(1,2)}~+{\bf s}^{(1,3)}~-{\bf s}^{(1,2)}~+{\bf s}^{(1,3)})
\end{array}
\right]
\end{eqnarray*}
\end{example}

\begin{figure}[htbp]
\begin{center}
\includegraphics[scale = .45]{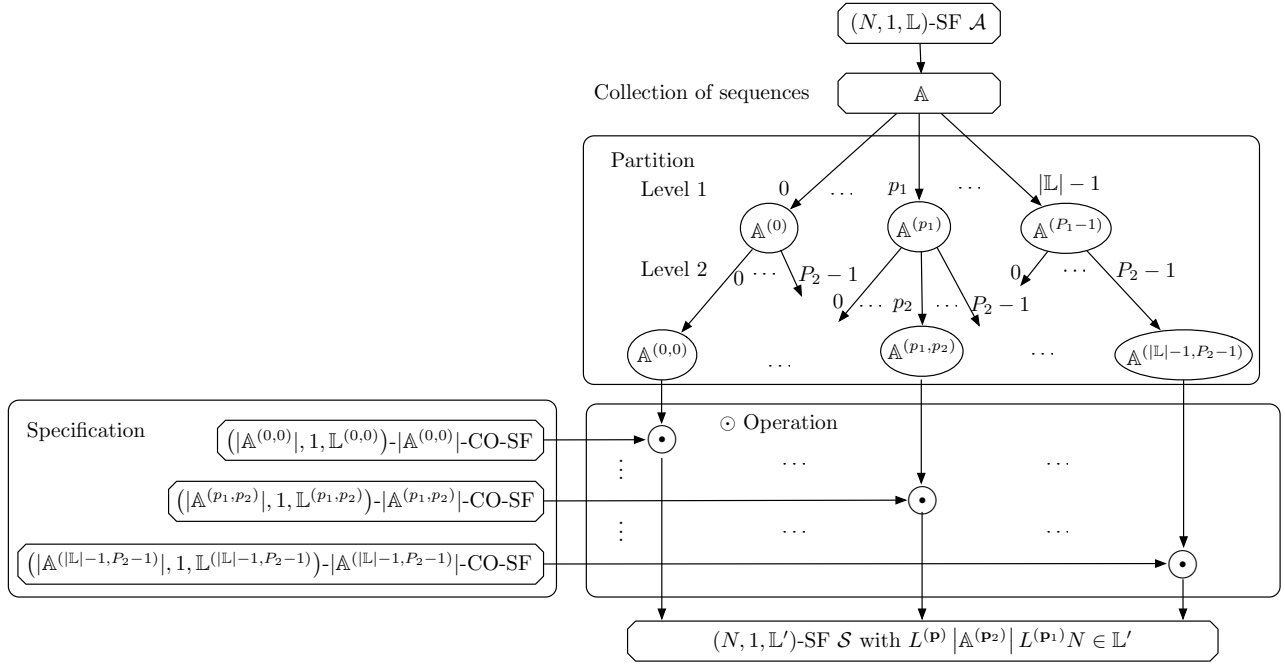}
\caption{Elongation algorithm of an $N$-CO-SF}
\label{Algorithm}
\end{center}
\end{figure}

In general, for a given $(N,1,\mathbb{L})$-$N$-CO-SF ${\cal A}$ with $\mathbb{L} = \{l^{(p_1)}N\}_{p_1=0}^{|\mathbb{L}|-1}$, the following algorithm that is illustrated in Fig. \ref{Algorithm} provides elongation of $\cal A$.

\begin{algorithm}
~
\begin{enumerate}
\item {\it Partition}: Let $\mathbb{A}$ be the set of all the sequences in ${\cal A}$. We consider a $2$-level partition of $\mathbb{A}$ such that, at the 1st level, the set $\mathbb{A}$ is partitioned to $|\mathbb{L}|$ SSs $\mathbb{A}^{({\bf p}_1)}$, ${\bf p}_1 \in \mathbb{P}_1 = \{(p_1)\}_{p_1=0}^{|\mathbb{L}|-1}$, according to sequence lengths and, at the second level, each $\mathbb{A}^{({\bf p}_1)}$, an $(|\mathbb{A}^{({\bf p}_1)}|,l^{({\bf p}_1)}N)$-SS, is further partitioned into $P_2^{({\bf p}_1)}$ $(|\mathbb{A}^{({\bf p}_2)}|,l^{({\bf p}_1)}N)$-SSs $\mathbb{A}^{({\bf p}_2)}, {\bf p}_2 \in \{({\bf p}_1,p_2);~0 \leq p_2<P_2^{({\bf p}_1)}\}$.

\item {\it Specification}: For each ${\bf p}_2 \in \mathbb{P}_2$, let $\mathcal{V}^{({\bf p}_2)}=\left\{\{{\bf v}^{({\bf p}_2),m}\}\right\}_{m=0}^{\left|\mathbb{A}^{({\bf p}_2)}\right|-1}$ be an $\left(|\mathbb{A}^{({\bf p}_2)}|,1,\mathbb{L}^{({\bf p}_2)}\right)$-$|\mathbb{A}^{({\bf p}_2)}|$-CO-SF, the $m$th SS of which has sequence length $l^{({\bf p}_2,m)}|\mathbb{A}^{({\bf p}_2)}|$ with a positive integer $l^{({\bf p}_2,m)}$.

\item {\it Connection operation}: Let $\mathbb{P}= \{({\bf p}_2,m);0\leq m < |\mathbb{A}^{({\bf p}_2)}|,{\bf p}_2 \in \mathbb{P}_2\}$ and let $\mathcal{S}$ be an $(N,1,\mathbb{L}^\prime)$-SF, $\mathbb{L}^\prime = \sqcup\{l^{({\bf p})}\left|\mathbb{A}^{({\bf p}_2)}\right|l^{({\bf p}_1)}N\}_{{\bf p} \in\mathbb{P}}$, given by 
\begin{eqnarray}
\mathcal{S}= \left\{\left\{{\bf s}^{({\bf p})}\right\}\right\}_{{\bf p} \in \mathbb{P}}= \left\{\left\{{\bf v}^{({\bf p}_2),m}\odot \mathbb{A}^{({\bf p}_2)}\right\}\right\}_{{\bf p} \in \mathbb{P}}
\label{NCOSexp}
\end{eqnarray}
where we used the convention that ${\bf p} = ({\bf p}_2,m) \in \mathbb{P}$ specifies ${\bf p}_2 \in \mathbb{P}_2$ for some $m$ and ${\bf p}_2 = ({\bf p}_1,p_2)$ specifies ${\bf p}_1 \in \mathbb{P}_1$ for some $p_2$.
\end{enumerate}
\end{algorithm}

The following theorem is proved in Appendix III.
\begin{theorem}
If all the sequences in $\mathbb{A}^{({\bf p}_2)}$ have the same energy for each ${\bf p}_2$, the $(N,1,\mathbb{L}^\prime)$-SF $\mathcal{S}$ constructed by Algorithm 2 is an $N$-CO-SF.
\end{theorem}

We note, for each subset at the $1$st partition level, that its partition at the $2$nd partition level can be arbitrary and that different partitions result in different $N$-CO-SFs. For example, if we assume partition $\mathbb{A}^{(0,0)} = \{{\bf s}^{(0,m)}\}_{m=0}^{1}$, $\mathbb{A}^{(1,0)} = \{{\bf s}^{(1,0)}\}$, and $\mathbb{A}^{(1,1)}=\{{\bf s}^{(1,m)}\}_{m=1}^{3}$ with, respectively, $N^{(0,0)}=2$, $N^{(1,0)}=1$, and $N^{(1,1)}=3$, and if we consider, correspondingly, $2$-CO-SF, $1$-CO-SF, and $3$-CO-SF be $\mathcal{V}^{(0,0)} = \left\{(++),(+-)\right\},~\mathcal{V}^{(1,0)}=\{+\}$, and  $\mathcal{V}^{(1,1)} = \left\{(W_3,++),(+W_3+),(++W_3)\right\}$, then the expression (\ref{NCOSexp}) gives another optimal $(6,1,\{24,72\})$-$6$-CO-SF
\begin{eqnarray*}
\mathcal{S}
 =
\left[
\begin{array}{l}
(+{\bf s}^{(0,0)}~+{\bf s}^{(0,1)})\cr
(+{\bf s}^{(0,0)}~-{\bf s}^{(0,1)})\cr
(+{\bf s}^{(1,0)})\cr
(W_3{\bf s}^{(1,1)}~+{\bf s}^{(1,2)}~+{\bf s}^{(1,3)})\cr
(+{\bf s}^{(1,1)}~W_3{\bf s}^{(1,2)}~+{\bf s}^{(1,3)})\cr
(+{\bf s}^{(1,1)}~+{\bf s}^{(1,2)}~W_3{\bf s}^{(1,3)})
\end{array}
\right]
\end{eqnarray*}

\section{complete complementary codes}

In this section, we study the least upper bound on the family size of a CCC as well as its generation and enlargement.

\subsection{Upper bound and generation algorithm for CCC}
As noted before, a CCC in a matrix form has a certain similarity to a unitary-like matrix.
The following theorem is parallel to the fact that every $M\times N$ matrix $\bf G$ with the property ${\bf G}{\bf G}^H = \alpha {\bf I}_M$ for $\alpha > 0$ satisfies $M \leq N$ and is proved in Appendix IV (see also \cite{Han05-09J}).
\begin{theorem}
Every $(M,N,\mathbb{L})$-CCC satisfies $M \leq N$.
\end{theorem}

We say that an $(M,N,\mathbb{L})$-CCC is optimal if $M = N$. A unitary-like matrix ${\bf U}_N$ is a special (or optimal) case of CCC satisfying $M=N$ and may be considered as an optimal CCC with length set $\mathbb{L} = \{1\}$.

We prove the following theorem in Appendix V.
\begin{theorem}
Given an $\left(N,1,\mathbb{L}\right)$-$N$-CO-SF $\mathcal{S}=\left\{\right\{{\bf s}^m\}\}_{m=0}^{N-1}$, let $\mathbb{S}^{(m)}$ be $(l^{(m)}N,1)$-SSs $\mathbb{S}^{(m)}= \{(s^m(l))\}_{l=0}^{l^{(m)}N-1}$, $0 \leq m < N$. Then,  for a unitary-like matrix ${\bf U}_N$, an $(N,N,\mathbb{L})$-SF constructed as
\begin{eqnarray}
\mathcal{C} = \left\{\left\{{\bf c}_n^m\right\}_{n=0}^{N-1}\right\}_{m=0}^{N-1} =  \left\{\left\{{\bf u}_N^{n}\odot\mathbb{S}^{(m)}\right\}_{n=0}^{N-1}\right\}_{m=0}^{N-1}
\label{CCCgen}
\end{eqnarray}
 is a CCC.
\end{theorem}

\begin{example}
From the $(2,1,\{4\})$-$2$-CO-SF $\cal S$ in Example 3, we generate two SSs $\mathbb{S}^{(0)} = \left\{(+),(+),(+),(-)\right\}$ and $\mathbb{S}^{(1)} = \left\{(+),(+),(-),(+)\right\}$. If we let ${\bf U}_2={\bf H}_2$, then we have an optimal $(2,2,\{4\})$-CCC
\begin{eqnarray}
\mathcal{C}=\left[\begin{array}{cc}
{\bf c}_0^0& {\bf c}_1^0 \cr
{\bf c}_0^1& {\bf c}_1^1 
\end{array}\right]=\left[
\begin{array}{cc}
(+++-) & (+-++) \cr
(++-+) & (+---)
\end{array}\right]
\label{CCCexam1}
\end{eqnarray}
\end{example}

For a pair of  length-$N$ vector ${\bf u}$ and ${\bf v}$, we introduce entry-wise multiplication ${\bf u}\cdot{\bf v} := \left(u_nv_n\right)_{n=0}^{N-1}$. Given a unitary matrix ${\bf U}_N$, if we consider the associated $(N,1,\{N\})$-$N$-CO-SF ${\cal S}=\left\{\{{\bf u}_N^n\}\right\}_{n=0}^{N-1}$ in Theorem 5, then ${\bf u}_N^n\odot \mathbb{S}^{(m)} = {\bf u}_N^n\cdot {\bf u}_N^m$ and we have the following Corollary 1.
\begin{corollary}
 Given a unitary-like matrix ${\bf U}_N$, the SF
\begin{eqnarray}
\mathcal{C} = \left\{\left\{{\bf c}_n^m\right\}_{n=0}^{N-1}\right\}_{m=0}^{N-1} =  \left\{\left\{{\bf u}_N^{m}\cdot{\bf u}_N^{n}\right\}_{n=0}^{N-1}\right\}_{m=0}^{N-1}
\label{CCCgenS}
\end{eqnarray}
is an $(N,N,\{N\})$-CCC.
\end{corollary}

Since the CCC given by Theorem 5 is constructed from an $N$-CO-SF, it has the restriction that the sequence lengths are lower bounded by $N$. To construct an $(N,N,\mathbb{L})$-CCC with $\max(\mathbb{L})  < N$, we propose enlargement of a given CCC.

\subsection{Enlargement of a CCC}
For a length-$M$ vector ${\bf v}$ and $(N,L)$-SS $\mathbb{C}$, we introduce an operator\footnote{This operation is defined by reference to the Kronecker product.} ${\bf v}\otimes\mathbb{C} := \left\{v_{[k]_M}{\bf c}_{\lfloor k/M\rfloor}\right\}_{k=0}^{MN-1}$ to construct an $(MN,L)$-SS. Then, the next theorem gives a method to enlarge a given CCC.
\begin{theorem}
Given an $(N,N,\mathbb{L})$-CCC ${\cal C}=\{\mathbb{C}^n\}_{n=0}^{N-1}$ and $N$ unitary-like matrices ${\bf U}_M^{(n)}=\left[{\bf u}_M^{(n),m}\right]_{m=0}^{M-1}$, $0 \leq n < N$, let $\mathbb{E}^{nM+m}$ be SSs given by
\begin{eqnarray}
\mathbb{E}^{nM+m}= {\bf u}_M^{(n),m}\otimes \mathbb{C}^n
\label{CCCexp}
\end{eqnarray}
for $0 \leq m < M$ and $0 \leq n < N$. Then, $\mathcal{E} = \left\{\mathbb{E}^{k}\right\}_{k=0}^{MN-1}$ is an $(MN,MN,\mathbb{L})$-CCC.
\end{theorem}

\begin{proof}
The proof comes from the equalities
\begin{eqnarray*}
{\cal R}_{\mathbb{E}^{nM+m},\mathbb{E}^{n^\prime M+m^\prime}}(\tau) = {\bf u}_M^{(n),m}\left[{\bf u}_M^{(n^\prime),m^\prime}\right]^H {\cal R}_{\mathbb{C}^{n},\mathbb{C}^{n^\prime}}(\tau) = E_{{\bf u}_M^{(n),m}}E_{\mathbb{C}^n}\delta(n-n^\prime)\delta(m-m^\prime)\delta(\tau)
\end{eqnarray*}
for $0 \leq n,n^\prime < N$ and $0 \leq m,m^\prime < M$.
\end{proof}

\begin{example}
As an example, let ${\bf U}_2^{(0)} = {\bf I}_2$ and ${\bf U}_2^{(1)} = {\bf H}_2$ and enlarge the CCC given in Example 6 according to Theorem 6. Then, we have following optimal $(4,4,\{4\})$-CCC given in a matrix form
\begin{eqnarray}
\mathcal{E}=\left[\begin{array}{cccc}
{\bf e}_0^0& {\bf e}_1^0 & {\bf e}_2^0 & {\bf e}_3^{0}\cr
{\bf e}_0^1& {\bf e}_1^1 & {\bf e}_2^1 & {\bf e}_3^{1}\cr
{\bf e}_0^2& {\bf e}_1^2 & {\bf e}_2^2 & {\bf e}_3^{2}\cr
{\bf e}_0^3& {\bf e}_1^3 & {\bf e}_2^3 & {\bf e}_3^{3}
\end{array}\right]=\left[\begin{array}{cccc}
(++-+) & (+---) & {\bf 0}_4 &{\bf 0}_4 \cr
{\bf 0}_4 & {\bf 0}_4 &(++-+) & (+---)\cr
(+++-)& (+-++) & (+++-) & (+-++) \cr
(+++-) & (+-++) & (---+) & (-+--)
\end{array}\right]
\label{Eexp}
\end{eqnarray} 

A CCC with $MN > \max( \mathbb{L})$ is also possible by choosing large unitary-like matrices. For example, if we let ${\bf U}_4 = {\bf H}_4$, then Theorem 6 gives an $(8,8,\{4\})$-CCC.
\begin{eqnarray*}
\mathcal{E}=\left[\begin{array}{cccccccc}
+{\bf c}_0^0 & +{\bf c}_1^0 & +{\bf c}_0^0 & +{\bf c}_1^0 & +{\bf c}_0^0 & +{\bf c}_1^0 & +{\bf c}_0^0 & +{\bf c}_1^0 \cr
+{\bf c}_0^1 & +{\bf c}_1^1 & +{\bf c}_0^1 & +{\bf c}_1^1 & +{\bf c}_0^1 & +{\bf c}_1^1 & +{\bf c}_0^1 & +{\bf c}_1^1 \cr
+{\bf c}_0^0 & +{\bf c}_1^0 & -{\bf c}_0^0 & -{\bf c}_1^0 & +{\bf c}_0^0 & +{\bf c}_1^0 & -{\bf c}_0^0 & -{\bf c}_1^0 \cr
+{\bf c}_0^1 & +{\bf c}_1^1 & -{\bf c}_0^1 & -{\bf c}_1^1 & +{\bf c}_0^1 & +{\bf c}_1^1 & -{\bf c}_0^1 & -{\bf c}_1^1 \cr
+{\bf c}_0^0 & +{\bf c}_1^0 & +{\bf c}_0^0 & +{\bf c}_1^0 & -{\bf c}_0^0 & -{\bf c}_1^0 & -{\bf c}_0^0 & -{\bf c}_1^0 \cr
+{\bf c}_0^1 & +{\bf c}_1^1 & +{\bf c}_0^1 & +{\bf c}_1^1 & -{\bf c}_0^1 & -{\bf c}_1^1 & -{\bf c}_0^1 & -{\bf c}_1^1 \cr
+{\bf c}_0^0 & +{\bf c}_1^0 & -{\bf c}_0^0 & -{\bf c}_1^0 & -{\bf c}_0^0 & -{\bf c}_1^0 & +{\bf c}_0^0 & +{\bf c}_1^0 \cr
+{\bf c}_0^1 & +{\bf c}_1^1 & -{\bf c}_0^1 & -{\bf c}_1^1 & -{\bf c}_0^1 & -{\bf c}_1^1 & +{\bf c}_0^1 & +{\bf c}_1^1 \cr
\end{array}\right]
\label{Eexp}
\end{eqnarray*}

\end{example}
\section{Family size, alphabet size, and sequence length}

The CCC has attracted the attentions of numerous researchers because of its ideal correlation properties. Besides correlation properties, however, family size, alphabet size, and sequence length also play important roles in practical applications. In this section, we discuss these factors of the CCCs constructed by our algorithms.

\subsection{Family size}
In Theorem 5, we gave a method to generate a CCC from an $N$-CO-SF and $N\times N$ unitary-like matrix ${\bf U}_N$. As a result, the generated CCC has, in matrix form, rows consisting of CSs and columns consisting of $N$-CO-SFs, and hence the family size of the CCC and the size of each of the component CSs depend on the family size of the employed $N$-CO-SFs. On the other hand, Theorem 4 shows that the family size of the CCC is upper bounded by the size of the component CSs. The two facts imply that the only way to increase the family size of the CCC is to increase the set size of the component CSs, and Theorem 6 gives a method to increase the set size of the component CSs. We note that the enlargement in Theorem 6 does not change the length set.

\subsection{Alphabet size}
In Theorems 2, 3, 5, and 6, we only consider unitary-like matrices and the operations `connection' and `multiplication'. Thus, if we can show the entries of the resultant sequences are bounded in a finite set, finite-alphabet construction becomes possible. Unitary-like matrices which allow this property are ${\bf F}_N$, which exists for any positive integer $N$, and ${\bf H}_N$, which exists for $N=2^m$ with any positive integer $m$. Employing ${\bf F}_N$ in Corollary 1, we can derive a polyphase CCC ${\cal C} = \left\{\{{\bf f}_N^{nm}\}_{n=0}^{N-1}\right\}_{m=0}^{M-1}$ with alphabet size $N$ while a binary CCC ${\cal C} = \left\{\{{\bf h}_N^{n\oplus m}\}_{n=0}^{N-1}\right\}_{m=0}^{M-1}$ can be derived by employing ${\bf H}_N$  in Corollary 1, where $n\oplus m$ denotes the dyadic summation\footnote{The dyadic summation of $n$ and $m$ is equal to $t$ if and only if their binary representations $n = (n_0,n_1,\cdots,n_{I-1})$, $m = (m_0,m_1,\cdots,m_{I-1})$, and $t = (t_0,t_1,\cdots,t_{I-1})$ satisfy $[n_i+m_i]_2=t_i$, for all $0 \leq i < I$.} of $n$ and $m$. For example, employing ${\bf H}_4$ in Corollary 1, we have a $(4,4,\{4\})$-CCC as
\begin{eqnarray}
{\cal C}=\left[\begin{array}{cccc}
{\bf c}_0^0& {\bf c}_1^0 & {\bf c}_2^0 & {\bf c}_3^{0}\cr
{\bf c}_0^1& {\bf c}_1^1 & {\bf c}_2^1 & {\bf c}_3^{1}\cr
{\bf c}_0^2& {\bf c}_1^2 & {\bf c}_2^2 & {\bf c}_3^{2}\cr
{\bf c}_0^3& {\bf c}_1^3 & {\bf c}_2^3 & {\bf c}_3^{3}
\end{array}\right]=\left[\begin{array}{cccc}
{\bf h}_4^0& {\bf h}_4^1 & {\bf h}_4^2 & {\bf h}_4^{3}\cr
{\bf h}_4^1& {\bf h}_4^0 & {\bf h}_4^3& {\bf h}_4^{2}\cr
{\bf h}_4^2& {\bf h}_4^3 & {\bf h}_4^0 & {\bf h}_4^{1}\cr
{\bf h}_4^3& {\bf h}_4^2& {\bf h}_4^1& {\bf h}_4^{0}
\end{array}\right]
\end{eqnarray}

For an $(N,N,\{L\})$-CCC, in some applications as CCC-CDMA, it is expected for a higher spectral efficiency that the sum of correlations with adjacent sequence also vanishes as
\begin{eqnarray}
\sum_{n=0}^{N-1} R_{{\bf c}_{[n+1]_N}^m,{\bf c}_{n}^{m^\prime}}(L-\tau) = 0~{\rm if}~0 < \tau \leq Z
\label{ZCCC}
\end{eqnarray}
for certain positive integer $Z$. Actually, if the SS generated by connecting sequences in each CS of a CCC becomes ZCZ-SS,  then (\ref{ZCCC}) holds. Hence,  in \cite{Han09-06C}, we called such a CCC a {\it Z-connectable CCC} (Z-CCC) and proved that the $(N,N,\{N\})$-CCCs derived by employing ${\bf F}_N$ and ${\bf H}_N$ in Corollary 1 are Z-CCC with $Z=N-1$ and $Z = N/2$, respectively.
 
\subsection{Sequence length}
Although each construction shown in this paper has a certain restriction on the length of resultant sequences, we can construct quite a large class of CCCs of variety of lengths by combining the proposed algorithms.

We first consider constructible sequence lengths for $N$-CO-SFs. In our framework, one can generate an initial $N$-CO-SF by Theorem 2 and extends the result using Theorem 3 iteratively. On the other hand, we may identify the collection of rows of a unitary-like matrix ${\bf U}_N$ with an $(N,1,\{N\})$-$N$-CO-SF. Thus, we can show the following theorem by mathematical induction.
\begin{theorem}
Each sequence in the $N$-CO-SFs constructed by Algorithm 1 and by iterative application of Algorithm 2 has a length equal to a product of integers which are not greater than $N$. Conversely, given an integer $L$ which is decomposed into $N$ and factors not greater than $N$,  an $N$-CO-SF whose length set includes $L$ can be constructed. 
\end{theorem}

\begin{proof}
For an $N$, Algorithm 1 gives $(N,1,\mathbb{L})$-$N$-CO-SFs with length set $|\mathbb{A}^{({\bf p}_1)}|N \in \mathbb{L}$ for subsets $\mathbb{A}^{({\bf p}_1)}$ of $\mathbb{A}$. Therefore, $|\mathbb{A}^{({\bf p}_1)}| \leq |\mathbb{A}| =N$ and the first half of Theorem 7 is true for the $N$-CO-SFs constructed by Algorithm 1.

Next, let us consider the lengths of $N$-CO-SFs constructed by iterative application of Algorithm 2. We assume that we are given an $(N,1,\mathbb{L})$-$N$-CO-SF with $l^{({\bf p}_1)}N \in \mathbb{L}$, at the beginning of Algorithm 2, and an $(|\mathbb{A}^{({\bf p}_2)}|,1,\mathbb{L}^{({\bf p}_2)})$-$|\mathbb{A}^{({\bf p}_2)}|$-CO-SF ${\cal V}^{({\bf p}_2)}$ with $l^{({\bf p})}|\mathbb{A}^{({\bf p}_2)}| \in \mathbb{L}^{({\bf p}_2)}$, at Specification step of Algorithm 2, where $l^{({\bf p}_1)}$ and $l^{({\bf p})}$ are assumed to be decomposed into factors which are not greater than $N$ and factors which are not greater than $|\mathbb{A}^{({\bf p}_2)}|$, respectively. Then, the connection operation ${\bf v}^{({\bf p}_2),m}\odot \mathbb{A}^{({\bf p}_2)}$ gives an SS with $l^{({\bf p})}\left|\mathbb{A}^{({\bf p}_2)}\right|l^{({\bf p}_1)}N \in \mathbb{L}^\prime$. Obviously, $\left|\mathbb{A}^{({\bf p}_2)}\right| \leq N$. Thus, given constituent $N$-CO-SFs satisfying the first half of the theorem, Algorithm 2 gives an $N$-CO-SF which also satisfies it. 

Conversely, if a length can be decomposed by factors which are not greater than $N$, the following algorithm yields an $N$-CO-SF whose length set includes such a length. Assume the target length is $L = N\prod_{j=0}^{J-1}l^{(j)}$ for $N \geq l^{(0)} \geq l^{(1)},\cdots,\geq  l^{(J-1)} > 1$, then an $N$-CO-SF whose length set includes $l^{(0)}N$ can be constructed by Algorithm 1 from unitary-like matrices ${\bf U}_N$ and ${\bf U}_{l^{(0)}}$ and the resultant $N$-CO-SF includes at least $l^{(0)}$ sequences with length $Nl^{(0)}$. Hence, at Partition step of Algorithm 2, we have $|\mathbb{A}^{({\bf p}_1)}| \geq l^{(0)}$ and may select $\mathbb{A}^{({\bf p}_2)}$ such that $|\mathbb{A}^{({\bf p}_2)}| = l^{(1)}$. By specifying $(l^{(1)},1,\{l^{(1)}\})$-$l^{(1)}$-CO-SF consisting of rows of $l^{(1)} \times l^{(1)}$ unitary-like matrix at Specification step, Algorithm 2 yields at least $l^{(1)}$  length-$Nl^{(0)}l^{(1)}$ sequences and, from $l^{(1)} \geq l^{(2)},\cdots,\geq  l^{(J-1)} > 1$, we may select $\mathbb{A}^{({\bf p}_2)}$ such that $|\mathbb{A}^{({\bf p}_2)}| = l^{(2)}$ at Partition step in the next iteration. By repeat the above process, we can construct an $N$-CO-SF which includes length $N\prod_{j=0}^{J-1}l^{(j)}$ sequence(s). This completes the proof.
\end{proof}

Since $1$ and $2$ are only two integers not greater than $2$, for instance, a $2$-CO-SF whose length set includes $2^n$ can be constructed for any positive integer $n$ while a length-$6$ $2$-CO-SF can not be constructed with our algorithms since $6$ includes a factor $3$ which is greater than $2$. However, length-$6$ $2$-CO-SF and other $N$-CO-SFs which can not be constructed by our algorithm  are not found at this moment.

Theorem 5 gives a mapping from a pair of an $(N,1,\mathbb{L})$-$N$-CO-SF and a unitary-like matrix ${\bf U}_N$ to an $(N,N,\mathbb{L})$-CCC and lays a bridge between $N$-CO-SFs and CCCs of the same length set $\mathbb{L}$. If we need an $(N,N,\mathbb{L})$-CCC with large sequence lengths for a fixed $N$, we can extend an $N$-CO-SF by Algorithm 2 and apply the operation in Theorem 5. When we need to construct an $(N,N,\mathbb{L})$-CCC with short sequence lengths, on the contrary, we construct $(N^\prime,N^\prime,\mathbb{L})$-CCC based on $N^\prime$-CO-SF with $N^\prime < N$ and enlarge the derived CCC by the operation in Theorem 6 to achieve the family size $N$. As the results, we can construct CCCs whose lengths are independent with $N$. Combining the result on the constructible lengths of $N$-CO-SF, the CCCs constructed in our framework can be $(\prod_{i=0}^{I-1}M^{(i)}N,\prod_{i=0}^{I-1}M^{(i)}N,\mathbb{L})$-CCCs whose length sets may contain length $N\prod_{j=0}^{J-1}l^{(j)}$ for arbitrary $1\leq l^{(j)} \leq N$, where $M^{(i)}$ are any positive integers.
Naturally, a question arises whether all CCCs are composed of $N$-CO-SFs? Untill now, we have not found a CCC which is not composed of $N$-CO-SFs, but its proof is not yet fully substantiated. 

Existing construction algorithms for CCCs can be understood in our framework. For example, CCCs with length $N^2$ proposed by Suehiro and Hatori \cite{Suehiro88-01J} can be generated for $P_1 = 1$ in  Algorithm 1 and the lengthening method described in \cite{Suehiro88-01J} can be obtained for $P_2 = 1$ in Algorithm 2. Han's construction method in \cite{Han04-04J} can be obtained for $P_1=N/M$ with $|\mathbb{A}^{({\bf p}_1)}|=M$ for all ${\bf p}_1$. The method given in \cite{Yang07-11C} can be considered as application of Theorem 6 to a $(2,2,\{L\})$-CCC $\cal C$ and ${\bf U}_M^{(n)} = {\bf U}_N$ for all $n$.

\section{Conclusion}
In this paper, we proposed systematic and optimal constructions of $N$-CO-SFs and CCCs. These constructions are realized by generation and extension methods and only unitary-like matrices are used. For any positive integers $M^{(i)}$, $0 \leq i < I$, an $(\prod_{i=0}^{I-1}M^{(i)}N,\prod_{i=0}^{I-1}M^{(i)}N,\mathbb{L})$-CCC  with $\mathbb{L}$ consisting of lengths $N\prod_{j=0}^{J-1}l^{(j)}$ for $1\leq l^{(j)} \leq N$ can be constructed by the use of presented methods and this form of CCC covers all the existing CCCs. In our framework, moreover, the alphabet size may be controlled by the appropriate selection of unitary-like matrices.

\appendices
\section{Proof of Theorem 1}
We consider an arbitrary $(M,1,\mathbb{L})$-$N$-CO-SF $\mathcal{S}=\{\{{\bf s}^m\}\}_{m=0}^{M-1}$ with $L^{(m)}=l^{(m)}N$ and let $l := \max (\mathbb{L})/N$, $\hat{\bf s}^m = ({\bf s}^m, {\bf 0}_{(l-l^{(m)})N})$. Next we introduce an $lM \times lN$ matrix ${\bf G} = \left[{\bf g}^m\right]_{m=0}^{lM-1}= \left[T^{[m]_lN}(\hat{\bf s}^{\lfloor m/l \rfloor})\right]_{m=0}^{lM-1}$ where $T^{vN}(\hat{\bf s}^u)$ denotes cyclicly right $vN$ shift of $\hat{\bf s}^u$. Then, the $(m,m^\prime)$th entry of $\Phib={\bf G} {\bf G}^H$ is given by
\begin{eqnarray*}
\phi_{m^\prime}^m & = & 
= T^{[m]_lN}(\hat{\bf s}^{\lfloor m/l \rfloor}) \cdot \left( T^{[m^\prime]_lN}(\hat{\bf s}^{\lfloor m^\prime/l \rfloor}))\right)^H
=\tilde{R}_{\hat{\bf s}^{\lfloor m/l \rfloor},\hat{\bf s}^{{\lfloor m^\prime/l \rfloor}}}([m^\prime-m]_lN)
\end{eqnarray*}
and, from the definition of $N$-CO-SF and the identity (\ref{relcorr}), we have $\Phib = E_{{\bf s}^m} {\bf I}_{Ml}$. On the other hand, the fact that $rank(\Phib) \leq rank({\bf G})$ means $Ml \leq \min\{Nl,Ml\}$ and hence means the bound $M \leq N$.

\section{Proof of Theorem 2}
Let $l^{({\bf p}_1)} = |\mathbb{A}^{({\bf p}_1)}|$ and $\mathbb{A}^{({\bf p}_1)} = \left\{{\bf a}_i^{({\bf p}_1)}\right\}_{i=0}^{l^{({\bf p}_1)-1}}$for ${\bf p}_1 \in \mathbb{P}_1$. Because of Theorem 1, we only need to show that $\cal S$ is an $N$-CO-SF. Then, the correlation between ${\bf s}^{\bf p}$, ${\bf p} = ({\bf p}_1,m)$, and ${\bf s}^{{\bf p}^\prime}$, ${\bf p}^\prime = ({\bf p}_1^\prime,m^\prime)$, is given by
\begin{eqnarray*}
R_{ {\bf s}^{{\bf p}}, {\bf s}^{{\bf p}^\prime}} ( kN ) 
 &=&  \sum_{i=0}^{l-1}\sum_{j=0}^{N-1}s^{{\bf p}}(iN+j)\left\{s^{{\bf p}^\prime}\left([i+k]N+j\right)\right\}^\ast\\
& = & \sum_{i=0}^{l-1}u^{({\bf p}_1),m}(i)\left[u^{({\bf p}_1^\prime),m^\prime}\left([i+k]_{l^\prime}\right)\right]^\ast{\bf a}_i^{({\bf p}_1)}\left[{\bf a}_{[i+k]_{l^\prime}}^{({\bf p}_1^\prime)}\right]^H
\end{eqnarray*}
where we let $l^{({\bf p}_1)}=l$ and $l^{({\bf p}_1^\prime)}=l^\prime$ for simplicity.

Since $\mathbb{A}$ consists of rows of ${\bf U}_N$ and ${\bf U}_{l^{({\bf p}_1)}}^{({\bf p}_1)}$ which are unitary-like matrices for all ${\bf p}_1$, the correlation is calculated as
\begin{eqnarray*}
R_{ {\bf s}^{{\bf p}}, {\bf s}^{{\bf p}^\prime}} ( kN )& = & {\bf u}^{({\bf p}_1),m}\left[{\bf u}^{({\bf p}_1^\prime,m^\prime)}\right]^H \delta({\bf p}_1-{\bf p}_1^\prime)\delta(k) = E_{{\bf s}^{{\bf p}}}\delta({\bf p}-{\bf p}^\prime)\delta(k)
\end{eqnarray*}
This shows that $\cal S$ is an $N$-CO-SF and completes the proof.
\section{Proof of Theorem 3}
For a given ${\bf p} \in \mathbb{P}$, we let $l = l^{({\bf p})} = l^{({\bf p}_2,m)}$, $l_1 = l^{({\bf p}_1)}$, and $k = |\mathbb{A}^{({\bf p}_2)}|$, and, for another ${\bf p}^\prime \in \mathbb{P}$, we introduce $l^\prime$, $l^\prime_1$, and $k^\prime$ in a similar manner. Moreover, we let $\mathbb{A}^{({\bf p}_2)}=\left\{{\bf a}_n^{({\bf p}_2)}\right\}_{n=0}^{k-1}$.

We first note
\begin{eqnarray*}
{\bf s}^{({\bf p})} = {\bf v}^{({\bf p}_2),m}\odot \mathbb{A}^{({\bf p}_2)} = \sum_{i=0}^{lk-1} v^{({\bf p}_2),m}(i){\bf b}^{({\bf p}_2,i)}
\end{eqnarray*}
for ${\bf b}^{({\bf p}_2,i)} = \left({\bf 0}_{il_1N}~{\bf a}_{[i]_k}^{({\bf p}_2)}~{\bf 0}_{(lk-i-1)l_1N}\right)$ and, due to ${\bf a}_{[i]_k}^{({\bf p}_2)}$ is chosen from an $N$-CO-SF, we have
\begin{eqnarray}
R_{{\bf b}^{({\bf p}_2,i)},{\bf b}^{({\bf p}_2^\prime,i^\prime)}}(xN) &=& R_{{\bf b}^{({\bf p}_2,i)},{\bf b}^{({\bf p}_2,i^\prime)}}(xN) \delta({\bf p}_2-{\bf p}_2^\prime)\nonumber\\
&=&R_{{\bf a}^{({\bf p}_2)}_{[i]_k},{\bf a}^{({\bf p}_2)}_{[i^\prime]_{k}}}\left([x+i l_1-i^\prime l_1]N\right)\delta({\bf p}_2-{\bf p}_2^\prime)\nonumber\\
\label{corra}
&=&E_{{\bf a}^{({\bf p}_2)}_{[i]_k}} \delta\left(x+il_1 -i^\prime l_1\right)\delta([i-i^\prime]_k)\delta({\bf p}_2-{\bf p}_2^\prime)
\end{eqnarray}

On the other hand, the correlation between ${\bf s}^{({\bf p})}$ and ${\bf s}^{({\bf p}^\prime)}$ can be expressed by
\begin{eqnarray*}
R_{ {\bf s}^{({\bf p})}, {\bf s}^{({\bf p}^\prime)}} ( xN )
 &=&  \sum_{i=0}^{lk-1}\sum_{i^\prime=0}^{l^\prime k^\prime-1}v^{({\bf p}_2),m}(i)\left[v^{({\bf p}_2^\prime),m^\prime}(i^\prime)\right]^\ast R_{{\bf b}^{({\bf p}_2,i)},{\bf b}^{({\bf p}_2^\prime,i^\prime)}}(xN)
\end{eqnarray*}
and, since $E_{{\bf a}^{({\bf p}_2)}_{[i]_k}}$ is independent of $[i]_k$ for a given ${\bf p}_2$, by substituting $(\ref{corra})$, the correlation can be further calculated as
\begin{eqnarray*}
&&  E_{{\bf a}^{({\bf p}_2)}_{[i]_k}}\sum_{i=0}^{lk-1}\sum_{i^\prime=0}^{l^\prime k-1}v^{({\bf p}_2),m}(i)\left[v^{({\bf p}_2),m^\prime}(i^\prime)\right]^\ast \delta\left(x+i l_1 -i^\prime l_1\right)\delta([i-i^\prime]_k)\delta({\bf p}_2-{\bf p}_2^\prime)\\
&=&  E_{{\bf a}^{({\bf p}_2)}_{[i]_k}}\sum_{i=0}^{lk-1}v^{({\bf p}_2),m}(i)\left[v^{({\bf p}_2),m^\prime}\left(i+ x/l_1\right)\right]^\ast \delta\left(\left[x\right]_{\lcm(k,l_1)}\right) \delta({\bf p}_2-{\bf p}_2^\prime)\\
&=&  E_{{\bf a}^{({\bf p}_2)}_{[i]_k}}R_{{\bf v}^{({\bf p}_2),m},{\bf v}^{({\bf p}_2),m^\prime}}( x/l_1) \delta \left(\left[x\right]_{\lcm(k,l_1)}\right) \delta({\bf p}_2-{\bf p}_2^\prime)\\
&=&  E_{{\bf a}^{({\bf p}_2)}_{[i]_k}}E_{{\bf v}^{({\bf p}_2),m}}\delta(x) \delta({\bf p}_2-{\bf p}_2^\prime)
\end{eqnarray*}
provided by ${\bf v}^{({\bf p}_2),m}$ and ${\bf v}^{({\bf p}_2),m^\prime}$ are chosen from a $k$-CO-SF.

\section{Proof of Theorem 4}
Let $L=\max(\mathbb{L})$. Similarly to Appendix I, we extend each sequence in the CCC to length $L$ by $\hat{\bf s}_{n}^m = \left({\bf s}_{n}^m, {\bf 0}_{L-L^{(m)}}\right)$ and consider an $ML\times NL$ matrix ${\bf G} = \left[{\bf g}_0^m~{\bf g}_1^m~\cdots~{\bf g}_{N-1}^m\right]_{m=0}^{M-1}$ for ${\bf g}_n^m =T^{[m]_L}(\hat{\bf s}_n^{\lfloor m/L \rfloor})$.  Then, we have $\Phib = {\bf G}{\bf G}^H = E_{{\bf s}^m} {\bf I}_{ML} \leq rank({\bf G}) =\min\{NL,ML\} $ and hence $M \leq N$.

\section{Proof of Theorem 5}
Let $\mathbb{C}^m = \{{\bf c}_n^m\}_{n=0}^{N-1}$. Then, from (\ref{CCCgen}), ${\cal R}_{\mathbb{C}^m,\mathbb{C}^{m^\prime}}(\tau) $ is given by
\begin{eqnarray*}
&  & \sum_{k=0}^{l^{(m)}N-1} s^{(m)}(k)\left[s^{(m^\prime)}\left(k+\tau\right)\right]^\ast \sum_{n=0}^{N-1}u^n([k]_N)\left\{u^n\left([k+\tau]_N\right)\right\}^\ast\\
& = & R_{{\bf s}^{(m)},{\bf s}^{(m^\prime)}}(\tau)E_{{\bf u}_N^n}\delta([\tau]_N)\\
& = & E_{{\bf s}^{(m)}}\delta(m-m^\prime)\delta(\tau)
\end{eqnarray*}
This completes the proof.
   

\end{document}